\documentclass[12pt]{article}

\usepackage{amsmath}
\usepackage{amssymb}
\usepackage{amsthm}
\usepackage{latexsym}
\usepackage{color}
\usepackage{graphicx}
\usepackage{color}
\usepackage{mathrsfs}

\usepackage[symbol]{footmisc}

\DeclareSymbolFont{calletters}{OMS}{cmsy}{m}{n}
\DeclareSymbolFontAlphabet{\mathcal}{calletters}

\newtheorem{Theorem}{Theorem}[part]
\newtheorem{Definition}{Definition}[part]
\newtheorem{Proposition}{Proposition}[part]

\newtheorem{Lemma}{Lemma}[part]

\newtheorem{Remark}{Remark}[part]

\newcommand{\nc}{\newcommand}
\nc{\esssup}{\mathop{\mathrm{ess\,sup}}}
\nc{\essinf}{\mathop{\mathrm{ess\,inf}}}
\nc{\argmax}{\mathop{\mathrm{arg\,max}}}

\def \P{\mathbb{P}}
\def \N{\mathbb{N}}

\def \R{\mathbb{R}}

\def \E{\mathbb{E}}

\def \Q{\mathbb{Q}}

\def \1{\mathds{1}}

\def \Ac{{\cal A}}

\def \Cc{{\cal C}}

\def \Ec{{\cal E}}
\def \Fc{{\cal F}}

\def \Pc{{\cal P}}

\def \l({{\left (}}
\def \r){{\right )}}
\def \l[{{\left [}}
\def \r]({{\right ]}}

\def\Acr{\mathscr{A}}
\def\Ccr{\mathscr{C}}

\def \dint{{\displaystyle \int}}

\DeclareMathOperator*{\argmin}{arg\,min}
\newcommand{\MBFigure}[6]{
$\left. \right.$ \\
\refstepcounter{figure}
\addcontentsline{lof}{figure}{\numberline{\thefigure}{\ignorespaces #5}}
\begin{center}
\begin{minipage}{#1cm}
\centerline{\includegraphics[width=#2cm,angle=#3]{#4}}
\begin{center}
\upshape{F\textsc{ig} \normal
\end{center}
size{\thefigure}. $-$} #5
\end{center}
\label{#6}
\end{minipage}
\end{center}
$\left. \right.$ \\}
\makeatletter
\let\@fnsymbol\@arabic
\makeatother
\title{Optimal investment and consumption under logarithmic utility and uncertainty model }
\author{Wahid {\sc Faidi}
\thanks{Department of Mathematics, College of Science and Humanities, Shaqra University,
Al Quwayiyah, Saudi Arabia }\;
 \thanks{University of Tunis El Manar, LAMSIN, Tunis, Tunisia} 
\\ e-mail:  \textcolor[rgb]{0.00,0.07,1.00}{faidiwahid@su.edu.sa}- \textcolor[rgb]{0.00,0.07,1.00}{wahid.faidi@lamsin.rnu.tn}
}
\begin{document}
\maketitle
\begin{abstract}
We study a robust utility maximization problem in the case of an incomplete market and logarithmic utility with general stochastic constraints, not necessarily convex. Our problem is equivalent to maximizing of nonlinear expected logarithmic utility. We characterize the optimal solution using quadratic BSDE.
\end{abstract}
\textbf{Keywords}: Backward stochastic differential equations-  g-expectation  -  g-martingale  - Logarithmic utility   -  Robust Utility . 
\section{Introduction}
Utility maximization represents an  important problem in financial mathematics. This is an optimal investment problem faced by an economic agent having the possibility of investing in a financial market over a finite period of time with fixed investment horizon $T$. The goal of the agent is to find an optimal portfolio that allows him to maximize his "welfare" at time $T$. The founding work of Von Neumann and Morgenstern \cite{VM44}, made it possible to represent the preferences of the investor by means of a function of utility U and a given probability measure $\mathbb{P}$ reflecting his views as follows
$$E_{\mathbb{P}}[U(X_T)].$$
Where $X_T$ is the wealth of investor at time $T$. 
The investor's problem then consists in solving the optimization
$$\sup_{\pi}E_{\mathbb{P}}[U(X_T)].$$
To solve this type of problem there are essentially two approaches:the dual approach \cite{SB09} and the BSDE approach \cite{IH2005}
\\
In reality, several scenarios are plausible and it is impossible to precisely identify $\P$. Therefore, we must take into account this ambiguity on the model also known as Knightian uncertainty. Knightian uncertainty studies have undergone enormous developments in both theory and applications. The work of  Maccheroni, Marinacci and Rustichini \cite{MMR06} led to a new representation of preferences in the presence of model uncertainty, namely:
$$\inf_{Q \in \mathcal{Q}}E_{\mathbb{P}}[U(X_T)+\gamma(Q)],$$
where $\mathcal{Q}$ the set of plausible scenarios and $\gamma (Q)$ is the penalty term.In other words, the investor will decide in the worst case. Thus, it will solve the following optimization problem known as robust utility maximization problem:
$$\sup_{\pi}\inf_{Q \in \mathcal{Q}}E_{Q}[U(X_T)+\gamma(Q)].$$
Several developments have been implemented on this subject, whether using the duality method \cite{FG06,Que04,Sch07} or stochastic control techniques based on backward stochastic differential equations \cite{BMS05,FMM17}.
Using duality technical combined a PDE approach to the dual problem, Hernández and Schied \cite{HerSchied07} succeeded in characterizing the optimal value function using HJB equation.
In this paper  we studies the robust utility maximization problem in incomplete market setting. We characterize the value function of our optimization problem using the quadratic BSDE. This was worked out through the class of nonlinear expectation called $g^*$- expectation.
\\
This paper is organized as follows: Section 2 provides the problem setting, the necessary
notations, conceptions and some properties about the $g$-martingales.section 3 is devoted to the notion of $g-$ expectation.
In section 4, we specify the financial market and the set of admissible strategies. 
In section 5, we  characterize the optimal investment-consumption strategy by using backward stochastic differential equation.Finally we treat an example corresponding to the case of entropic penalty.
\section{Formulation Problem}
We consider a filtered space $(\Omega, \mathcal{F}, \mathbb{F}^W, \mathbb{P})$ over a finite horizon time $T$, where the filtration $\mathbb{F}^W=(\mathcal{F}^W_t)_{t \in [0,T]}$ is generated by standard $d$-dimensional  Brownian motion $W=(W^1,\ldots, W^d).$ 
In order to present our problem rigorously, we introduce some process spaces.
\begin{itemize}
\item[$\bullet$]For $(n,k)\in \N\times \N; \Pc^{n \times k}$ is the space of all predictable processes  with values in  $\R^{n\times k}$.
\\ $\Pc^{1 \times 1}$ will simply be noted $\Pc.$
\item[$\bullet$] For $p \in \mathbb{N}, \mathscr{H}_T^p\left(\mathbb{R}^m\right)$ is the set of all $\mathbb{R}^m$-valued stochastic processes $Z$ which are predictable with respect to $\mathbb{F}$ and satisfy $\E_P\left[\left(\dint_0^T\left|Z_t\right|^2 d t\right)^{\frac{p}{2}}\right]<\infty.$

\item[$\bullet$]
 $D_0^{\exp }$ is the space of all progressively measurable processes $y=$ $\left(y_t\right)_{0 \leq t \leq T}$ with
$$
\E_\P\left[\exp \left(\gamma \underset{0 \leq t \leq T}{\operatorname{ess\;\;sup}}\left|y_t\right|\right)\right]<\infty \quad \text { for all } \gamma>0.
$$
\item[$\bullet$]
$D_1^{\exp }$ denotes the space of all progressively measurable processes $y=\left(y_t\right)_{0 \leq t \leq T}$ such that,
$$
\E_\P\left[\exp \left(\gamma \int_0^T\left|y_s\right| \mathrm{d} s\right)\right]<\infty \quad \text { for all } \gamma>0.
$$
\end{itemize}
It is well known, according to the result of exponential martingale representation, that for every measure $Q\ll \P$ on $\mathcal{F}_T$ there is a predictable process $(\eta_t)_{t \in [0,T]}$
such that $\displaystyle E_\P[\int_0^T\parallel \eta_t\parallel^2 dt]<+\infty \;\; Q.a.s$
 and the density process of $Q$ with respect to $\mathbb{P}$ is an RCLL
martingale $Z^Q =(Z^Q_t)_{t \in [0,T]}$ given by:
$$
Z_{t}^{Q}=\mathcal{E}\left(\int_{0}^{t} \eta_{u} d W_{u}\right) Q . a . s, \forall t \in[0, T],
$$
where $\mathcal{E}(M)_{t}=\exp \left(M_{t}-\frac{1}{2}\langle M\rangle_{t}\right)$ denotes the stochastic exponential of a continuous local martingale $M .$ We introduce a consistent time penalty  by:
$$
\gamma_{t}(Q)=E_{Q}\left[\int_{t}^{T} h\left(\eta_{s}\right) ds \mid \mathcal{F}_{t}\right],
$$
where $h: \mathbb{R}^{d} \rightarrow[0,+\infty]$ is a proper,  convex and lower semi-continuous function such that $h(0) \equiv 0 .$ We also assume that there are two positive constants $\kappa_{1}$ and $\kappa_{2}$ satisfying :
$$
h(x) \geq \kappa_{1}\|x\|^{2}-\kappa_{2}.
$$
Our optimization problem  is written as follows:
$$\sup_{(\pi,c) \in\mathscr{A}_{e} }   \inf_{Q^\eta \in \mathcal{Q}}\E\Big[\bar{\alpha}U(X_T^{x,\pi,c})+\alpha\int_0^T e^{-\int_0^s \delta_udu} u(c_s)ds +\int_{0}^{T}e^{-\int_0^s \delta_udu} h\left(\eta_{s}\right) ds\Big],$$
where $\mathcal{Q}$ the space of all probability measures $Q^\eta $ on $(\Omega, \mathcal{F})$ such that $Q^\eta \ll \mathbb{P}$ on $\mathcal{F}_T$ and $\gamma_0(Q^\eta)< +\infty$. $U$ and $u$ are the utility functions and $\mathscr{A}_{e}$ is  the set of admissible strategies which will be specified later.
Our problem is then made up of two optimization subproblems. The first, the infmum, is studied by Faidi Matoussi and Mnif \cite{FMM17}. They have proven under the exponential integrability condition of the random variables $U(X_T^{x,\pi,c})$ and $\dint_0^Tu(c_s)ds$ that the infimum is reached in a one unique probability measure $\Q^{*}$ which is equivalent  to $\P$ and they characterized the value process of the dynamical optimization problem using the quadratic BSDE. More precisely, under the following assumptions:
\begin{itemize}
\item[\textbf{(H1)}] The processes $\delta$ is uniformly bounded,
\item[\textbf{(H2)}] $\E_{\mathbb{P}}[\exp{\gamma |U(X_T^{x,\pi,c})|}]< +\infty \;\; \text{and} \;\; \E_{\mathbb{P}}[\exp{\gamma \int_0^T|u(c_s)|ds}]< +\infty \;\; \forall \gamma > 0,$
\end{itemize} 
  The process $$Y^{x,\pi,c}_t=\inf_{Q^\eta \in \mathcal{Q}}\E\Big[\bar{\alpha}U(X_T^{x,\pi,c})+\alpha\int_t^T e^{-\int_t^s \delta_udu} u(c_s)ds +\int_{t}^{T}e^{-\int_t^s \delta_udu} h\left(\eta_{s}\right) ds |\Fc_t\Big],$$
 satisfied the following quadratic BSDE:
\begin{equation}\label{bsdecaraterization}
\left\{\begin{array}{l}
d Y_{t}^{x,\pi,c}=\left(\delta_{t} Y_{t}^{x,\pi,c}-\alpha u(c_t)+\beta h^{*}
(\frac{1}{\beta}Z_{t})\right) dt-Z_{t} d W_{t}, \\
Y_{T}^{x,\pi,c}=\bar{\alpha} \bar{U}(X_T^{x,\pi,c}) .
\end{array}\right.
\end{equation}
Where $h^*$ the Legendre-Fenchel transform of $h$.
Thus, our problem is equivalent to find:
\begin{equation}\label{opt}
V_0(x)=\sup_{(\pi,c) \in \mathscr{A}_{e}}Y_0^{x,\pi,c}
\end{equation}
such that $Y$ satisfied (\ref{bsdecaraterization}).
To solve this problem, we use the notion of $g$-expectation introduced by Peng \cite{Peng97} .
We then start with the main results on this notion of nonlinear expectation.
\section{ $g$-expectation}
The $g$-expectations concept was firstly introduced by Peng in \cite{Peng97} in the case of lipschitz generator.  from which most basic material of this section is taken 
Let $\xi$ a strictly positive random variable such that $\E_\mathbb{P}[\exp(\gamma |\xi|)]<\infty$ forall $\gamma>0$ and $g:\mathbb{R}\times \mathbb{R}^d \rightarrow \mathbb{R}$ satisfies the following conditions:
\begin{itemize}
\item $\forall  t \in [0, T ], g(t,0)=0$
\item $\forall t \in [0, T ], z \mapsto g(t, z)$ is a continuous convex (or concave) function,
\item There are two positives constants $ \beta$ and $\gamma$ such that  $$ \forall (t,z)\in \R\times \R^d; |g(t,\ z)|\leq\beta|z|^{2}+\gamma .$$
\end{itemize}
According Briand and Hu \cite{BriHu08}, the BSDE
$$Y_{t}=\displaystyle \xi+\int_{t}^{T}g(s,\ Z_{s})ds-\int^{T}_tZ_{s}dW_{s},   
$$
admits   a  unique  solution  $(Y,\ Z)\in \mathscr{S}_T^{\infty}(\mathbb{R})\times \mathscr{H}_T^p\left(\mathbb{R}^m\right).$ 
\\
$Y_{t}$ is called conditional $g$-expectation of $\xi$ under $\mathcal{F}_t$ and and is noted 
$\mathcal{E}_{g}[\xi|\mathcal{F}_t]$ and $Y_{0}$ is called $g$ -expectation of $\xi$, denoted by $Y_{0}=\mathcal{E}_{g}(\xi)$.

\begin{Definition}
\begin{enumerate}
\item
A process $(X_{t})_t$ is called a $\mathrm{g}$-martingale, if for any $s \leq t$
$$\mathcal{E}_{g}[X_{t}|\mathcal{F}_s]=X_{s}$$
\item
A process $(X_{t})_t$ is called a $\mathrm{g}$-submartingale, if for any $s \leq t$
$$\mathcal{E}_{g}[X_{t}|\mathcal{F}_s]\geq X_{s}$$
\item
A process $(X_{t})_t$ is called a $\mathrm{g}$-supermartingale, if for any $s \leq t$
$$\mathcal{E}_{g}[X_{t}|\mathcal{F}_s]\leq X_{s}$$
\end{enumerate}
\end{Definition}
The following result is an immediate consequence of the comparison theorem
\begin{Proposition} \label{prop1}
If $g_1 \leq g_2$ then,
\begin{enumerate}
\item Any $\mathrm{g}_1$-submartingale is $\mathrm{g}_2$-submartingale.
\item Any $\mathrm{g}_2$-supermartingale is $\mathrm{g}_1$-supermartingale.
\end{enumerate}
\end{Proposition}
\begin{proof}
We prove the first item, the second in a similar way.
Let $(X_t)_t$ a $\mathrm{g}_1$-submartingale. For $0 \leq s \leq t$, we denote by $Y^{(1)}_{s,t}:=\mathcal{E}_{g_1}[X_{t}|\mathcal{F}_s]$ and $Y^{(2)}_{s,t}:=\mathcal{E}_{g_2}[X_{t}|\mathcal{F}_s].$ By definition, the processes $(Y^{(1)}_{s,t})_{0 \leq s \leq t}$ and $(Y^{(2)}_{s,t})_{0 \leq s \leq t}$ satisfies respectively the following BSDE:
$$Y^{(1)}_{s,t}=\displaystyle X_t+\int_{s}^{t}g_1(u,\ Z^{(1)}_{u})du-\int^{t}_sZ^{(1)}_{u}dW_{u},$$
and
$$Y^{(2)}_{s,t}=\displaystyle X_t+\int_{s}^{t}g_2(u,\ Z^{(2)}_{u})du-\int^{t}_sZ^{(2)}_{u}dW_{u}.$$
Using comparison theorem, we obtain $$Y^{(1)}_{s,t} \leq Y^{(2)}_{s,t}\;\; \mathbb{P}.a.s , \;\; \forall \;\;0 \leq s \leq t.$$
Therefore, $\mathcal{E}_{g_2}[X_{t}|\mathcal{F}_s] \geq X_s, \forall \;\;0 \leq s \leq t.$
\end{proof}
\begin{Lemma} Let $Y$ satisfying (\ref{bsdecaraterization}), then
$$Y_0^{x,\pi,c}=\Ec_g[e^{-\int_0^T \delta_u du}\bar{\alpha} \bar{U}(X_T^{x,\pi,c})+\int_0^T\alpha e^{-\int_0^s \delta_u du} u(c_s)ds],$$
where
$$g(w,t,z)=-\beta e^{-\int_0^t \delta_u du} h^{*}
\left(-\frac{1}{\beta}e^{\int_0^t \delta_u du}z\right).$$
\end{Lemma}
\begin{proof}
Consider the stochastic process $L^{x,\pi,c}$ defined by:
$$L^{x,\pi,c}_t:= e^{-\int_0^t \delta_u du}Y^{x,\pi,c}_t+\int_0^t\alpha e^{-\int_0^s \delta_u du} u(c_s)ds$$
The process $L^{x,\pi,c}$ satisfy the following BSDE:
\begin{equation}\label{gcaracterization}
\left\{\begin{array}{l}
d L_{t}^{x,\pi,c}=\beta e^{-\int_0^t \delta_u du} h^{*}
\left(-\frac{1}{\beta}e^{\int_0^t \delta_u du}Z_{t}\right) dt+Z_{t} d W_{t}, \\
L_{T}^{x,\pi,c}=e^{-\int_0^T \delta_u du}\bar{\alpha} \bar{U}(X_T^{x,\pi,c})+\int_0^T\alpha e^{-\int_0^s \delta_u du} u(c_s)ds .
\end{array}\right.
\end{equation}
In other words
$$L_{t}^{x,\pi,c}=\Ec_g\Big[ e^{-\int_0^T \delta_u du}\bar{\alpha} \bar{U}(X_T^{x,\pi,c})+\int_0^T\alpha e^{-\int_0^s \delta_u du} u(c_s)ds |\Fc_t\Big].$$
\end{proof}
\begin{Remark}
The conditions mentioned above are verified by the BSDE(\ref{gcaracterization}). Indeed:
\begin{enumerate}
\item Since $h(0)=0$ and $h(x)\geq 0$ for all $x\in \R,$ we have $h^*(0)=\sup\limits_{y \leq 0}-h(y)=0.$
\item $h$ is supposed to be continuous and convex function so $h^*$ is also.
\item  $ \forall x \in \R^d; 
|h(x)| \geq \kappa_{1}\|x\|^{2}-\kappa_{2},
$
then 
$ \forall x \in \R^d; 
|h^*(x)| \leq \dfrac{1}{2\kappa_{1}}\|x\|^{2}-\kappa_{2}.
$

\end{enumerate}
\end{Remark}
Thus our problem reduces to a utility maximization problem under the nonlinear expectation.
In the sequel we are interested in the second problem, the supermum, when the utility function is the logarithmic utility. We also assume that the discount factor $\delta$ is deterministic. We first specify the structure of the financial market as well as the set of admissible strategies.

\section{Market Model}
The financial market consists of one bond without interest rate  and $m \leq d$ stocks. In case $m<d$, we face an incomplete market. The price process
of stock $i$ evolves according to the equation:
$$
\begin{aligned}
&\frac{d S_{t}^{i}}{S_{t}^{i}}=b_{t}^{i} d t+\sigma_{t}^{i} d W_{t}, \quad i=1, \ldots, m
\end{aligned}
$$
where $b^{i}$ (resp. $\sigma^{i}=(\sigma^{i}_1, \ldots , \sigma^{i}_d)$ ) is an $\mathbb{R}$-valued (resp. $\mathbb{R}^{1 \times d}$ -valued) predictable uniformly bounded stochastic process. The volatility matrix $\sigma$ is the the $m \times d$- matrix whose the $i^{th}$ line  is  given by the vector $\sigma^{i}$ for i ranging from $1$ to $m$. We assume that:
\begin{itemize}
\item $\sigma=\left(\sigma^{i}_j\right)_{i=1, \ldots, d, j=1, \ldots, m}$ has full rank.
\item The matrix $\sigma \sigma^{t r}$ is uniformly elliptic i.e., there are two constants  $0 < \varepsilon< K ,$ such that   $\varepsilon I_{m} \leq \sigma \sigma^{t r} \leq  K I_{m}$ $\P$-a.s. 
\item The predictable $\mathbb{R}^{m}$-valued process
$$
\theta_{t}=\sigma_{t}^{t r}\left(\sigma_{t} \sigma_{t}^{t r}\right)^{-1} b_{t}, \quad t \in[0, T]
$$
is then also uniformly bounded. 
\end{itemize}
An economic agent invests in the financial market can consume part of his wealth at intermediate times. Let  $\pi^i_t, 1 \leq i \leq m,$ the proportion of investor's wealth invested in the $i^{th}$ risky stock $S^i$ and $c_t$ the consumed proportion  rate  at time $t$. 
We assume that $\left(\pi, c\right)$ belong to space 
$$
\mathscr{H}_T:=  \left\{\left(\pi, c \right)\in \Pc^{1\times m} \times \Pc \;\;
 \text {such that}\;\; c_t>0\;dt \times d\P .a.e \;\;\text {and}\;\; \dint_0^T c_t dt<\infty, \dint_0^T\left|\pi_t\right|^2 dt<\infty, \mathbb{P} \text {-a.s }\right\}.
$$
 By the self-financing condition, the investor's wealth $X_t^{x,\pi,c}$ at time $t$ , starting from the positive initial capital $x$, satisfies 
$$X_t^{x,\pi,c}=x+\sum\limits_{i=1}^m \int_0^t \pi_s^i\frac{dS^i_s}{S^i_s}-\int_0^t c_sX_s^{x,\pi,c} ds,$$
 
or, equivalently
$$
\dfrac{dX_t^{x,\pi,c}}{X_t^{x,\pi,c}}=\pi_t\sigma_t(dW_t+\theta_tdt)-c_tdt\;\;; X_0^{\pi,c}=x.
$$
The investor's wealth process can be write
\begin{equation}\label{eq0}
X_t^{x,\pi,c}=x\Ec(\pi\sigma\textbf{.} W^{\Q})_t \exp(-\int_0^tc_sds)>0,
\end{equation}
where $\Ec$ is the stochastic exponential and $W^{\Q}_t= W_t +\int_0^t \theta_sds.$
\\
Our goal is to select consumption and investment controls which maximize the finite horizon robust expected Logarithmic utility of consumption and terminal wealth.
To formulate consumption and investment constraints, let   $\Ac \times \Cc \in \Pc^{1\times m} \times \Pc $ consist of all pair $(\pi,c)\in \mathscr{H}_T$ such that:
$$\forall \gamma>0; \E[e^{\gamma (|\ln(X_T^{x,\pi,c})|}]< \infty \;\; \text{and} \;\; \E[e^{\int_0^T |\ln(c_sX_s^{x,\pi,c})|ds}]< \infty . $$

We also recall the following definitions relating to conditional analysis. For more details see cheridito et all \cite{cheridito10, cheridito11}. 
A subset $A$ of $\mathcal{P}^{1 \times k}$  is sequentially closed if it contains every process $a$ that is the $\nu \otimes \mathbb{P}$-a.e. limit of a sequence $\left(a^n\right)_{n \geq 1}$ of processes in $A$. We call it $\mathcal{P}$-stable if it contains $1_B a+1_{B^c} a^{\prime}$ for all $a, a^{\prime} \in A$ and every predictable set $B \subset[0, T] \times \Omega$. We say $A$ is $\mathcal{P}$-convex if it contains $\lambda a+(1-\lambda) a^{\prime}$ for all $a, a^{\prime} \in A$ and every process $\lambda \in \mathcal{P}$ with values in $[0,1]$. In the whole paper we work with the following
\begin{Definition}
The set $\Acr_e=\Acr\times \Ccr$ of   admissible strategy consists of   non empty sequentially closed,  $\Pc$-stable and $L^0-$ bounded subset  on   $\Ac \times \Cc$. 
\end{Definition}

\section{Optimal consumption–investment strategy}
%\subsection{Verification result}
To solve this problem, we adopt the same approach as Imkeller and Hu \cite{IH2005} and Jiang \cite{Jiang16}. 
Let us first start with the following verification lemma.
\begin{Lemma}\label{lemma1}
For any given  $(\pi,c) \in  \mathscr{A}_{e}$ , if there exists a RCLL adapted process $R^{x,\pi,c}$  such that 
 \begin{itemize}
 \item $\forall (\pi,c) \in  \Acr\times \Ccr  ; R_{T}^{x,\pi,c}:=\bar{\alpha}e^{-\int_0^T \delta_u du }\ln( X_{T}^{x,\pi,c})+\dint_0^T \alpha e^{-\int_0^u \delta_s ds } \ln (c_u X_{u}^{x,\pi,c})du,$ 
 \item $\forall (\pi,c) \in  \Acr\times \Ccr  ; R^{x,\pi,c}$ is a $ g$-supermartingale,
 \item  There exists $(\pi^{*},c^{*}) \in \Acr\times \Ccr$ such that $R^{x, \pi^{*},c^*}$ is a $g$-martingale.
 \\Then $(\pi^{*},c^{*})$ is an optimal strategy to Problem (\ref{opt}).
 \end{itemize}
\end{Lemma}
 \begin{proof}
  For given $t \in[0, T]$ and $(\hat{\pi},\hat{c}) \in \mathscr{A}_{e}$, we introduce $\tilde{\pi}_{u}=\pi_{u} I_{u \leq t}+\hat{\pi}_{u} I_{u>t}$, and $\tilde{c}_{u}=c_{u} I_{u \leq t}+\hat{c}_{u} I_{u>t}$ one can see that ($\tilde{\pi},\tilde{c}) \in \mathscr{A}_{e}$ and $X_{t}^{x,\tilde{\pi},\tilde{c}}=X_{t}^{x,\pi,c}$. Since $R^{x, \tilde{\pi},\tilde{c}}$ is a $g$-supermartingale, we have
 \begin{align*}
 R_{t}^{x,\pi,c} & \geq \Ec_g[\bar{\alpha}e^{-\int_0^T \delta_u du }\Big(\ln( X_{T}^{x,\tilde{\pi},\tilde{c}})\Big)+\dint_0^T \alpha e^{-\int_0^u \delta_s ds } \ln (\tilde{c}_u X_{u}^{x,\tilde{\pi},\tilde{c}})du|\Fc_t]
 \\&=\Ec_g\Big[\bar{\alpha}e^{-\int_0^T \delta_u du }\Big(\ln( X_{t}^{x,\pi,c}\dfrac{X_{T}^{x,\hat{\pi},\hat{c}}}{X_{t}^{x,\hat{\pi},\hat{c}}})\Big)+\dint_0^t \alpha e^{-\int_0^u \delta_s ds } \ln (c_u X_{u}^{x,\pi,c})du \\&
 +\dint_t^T \alpha e^{-\int_0^u \delta_s ds } \ln (\hat{c}_u  X_{t}^{x,\pi,c}\dfrac{X_{u}^{x,\hat{\pi},\hat{c}}}{X_{t}^{x,\hat{\pi},\hat{c}}})du|\Fc_t\Big].
 \end{align*}
 Thus,
 \begin{align*}
 R_{t}^{x,\pi,c} & \geq \sup\limits_{(\hat{\pi},\hat{c})}\Ec_g\Big[\bar{\alpha}e^{-\int_0^T \delta_u du }\ln( X_{t}^{x,\pi,c}\dfrac{X_{T}^{x,\hat{\pi},\hat{c}}}{X_{t}^{x,\hat{\pi},\hat{c}}})+\dint_0^t \alpha e^{-\int_0^u \delta_s ds } \ln (c_u X_{u}^{x,\pi,c})du \\&
 +\dint_t^T \alpha e^{-\int_0^u \delta_s ds } \ln (\hat{c}_u  X_{t}^{x,\pi,c}\dfrac{X_{u}^{x,\hat{\pi},\hat{c}}}{X_{t}^{x,\hat{\pi},\hat{c}}})du|\Fc_t\Big].
 \end{align*}
 Furthermore, as for $(\pi^*,c^*); $ $R_{t}^{ \pi^{*}}$ is a $g$ -martingale,
\begin{align*}
 R_{t}^{x,\pi^*,c^*} & = \Ec_g\Big[\bar{\alpha}e^{-\int_0^T \delta_u du }\ln( X_{t}^{x,\pi^*,c^*}\dfrac{X_{T}^{x,\pi^*,c^*}}{X_{t}^{x,\pi^*,c^*}})+\dint_0^t \alpha e^{-\int_0^u \delta_s ds } \ln (c_u X_{u}^{x,\pi^*,c^*})du \\&
 +\dint_t^T \alpha e^{-\int_0^u \delta_s ds } \ln (c^*_u  X_{t}^{x,\pi^*,c^*}\dfrac{X_{u}^{x,\pi^*,c^*}}{X_{t}^{x,\pi^*,c^*}})du|\Fc_t\Big]
 \\ & \leq
  \sup\limits_{(\hat{\pi},\hat{c})}\Ec_g\Big[\bar{\alpha}e^{-\int_0^T \delta_u du }\ln( X_{t}^{x,\pi^*,c^*}\dfrac{X_{T}^{x,\hat{\pi},\hat{c}}}{X_{t}^{x,\hat{\pi},\hat{c}}})+\dint_0^t \alpha e^{-\int_0^u \delta_s ds } \ln (c_u X_{u}^{x,\pi^*,c^*})du \\&
 +\dint_t^T \alpha e^{-\int_0^u \delta_s ds } \ln (\hat{c}_u  X_{t}^{x,\pi^*,c^*}\dfrac{X_{u}^{x,\hat{\pi},\hat{c}}}{X_{t}^{x,\hat{\pi},\hat{c}}})du|\Fc_t\Big].
 \end{align*}
 Thus
 \begin{align*}
 R_{t}^{x,\pi^*,c^*} & = 
  \sup\limits_{(\hat{\pi},\hat{c})}\Ec_g\Big[\bar{\alpha}e^{-\int_0^T \delta_u du }\ln( X_{t}^{x,\pi^*,c^*}\dfrac{X_{T}^{x,\hat{\pi},\hat{c}}}{X_{t}^{x,\hat{\pi},\hat{c}}})+\dint_0^t \alpha e^{-\int_0^u \delta_s ds } \ln (c_u X_{u}^{x,\pi^*,c^*})du \\&
 +\dint_t^T \alpha e^{-\int_0^u \delta_s ds } \ln (\hat{c}_u  X_{t}^{x,\pi^*,c^*}\dfrac{X_{u}^{x,\hat{\pi},\hat{c}}}{X_{t}^{x,\hat{\pi},\hat{c}}})du|\Fc_t\Big],
 \end{align*}
 holds for each $t \in [0, T]$, and by taking $t = 0$, we know that $(\pi^*,c^*)$
 is an optimal strategy.
 \end{proof}
% \subsection{Deterministic discounting factor case}
 In the sequel, we assume that the process $\delta$ is deterministic. 
 \begin{Theorem}\label{Th1}
The optimal value of the problem (\ref{opt}) is given by
$$V_0(x)=(\bar{\alpha} e^{-\int_0^T \delta_u du}+ \alpha \int_0^T e^{-\int_0^s \delta_u du}ds)\ln(x)+Y_0,$$
where $(Y,Z)$ is the solution of the following BSDE
\begin{equation}\label{BSDE}
dY_t= -f(t,Z_t)dt+Z_t dW_t\;\;;\; Y_T=0,
\end{equation}
with
\begin{equation}\label{eq4}
\begin{split}
 f(t,z) = \esssup\limits_{\pi \in  \Acr }(g(t,z+\ell_t\pi_t \sigma_{t})-\ell_t \pi_t \sigma_{t}\theta_{t}) + \esssup\limits_{c \in  \Ccr }(\alpha e^{-\int_0^t \delta_s ds } \ln (c_t)-\ell_tc_t)
\;\; \P.a.s.
 \end{split}
\end{equation}
and 
\begin{equation}\label{ell}
\ell_t = \bar{\alpha} e^{-\int_0^T \delta_u du}+ \alpha \int_t^T e^{-\int_0^s \delta_u du}ds.
\end{equation}
Morover $(\pi^*,c^*)$  is an optimal admissible strategy if and only if
$$c^* \in \argmax\limits_{c\in \Ccr} \Big(\alpha e^{-\int_0^. \delta_s ds } \ln (c)-\ell c\Big) \;\;\; \text{and}  \;\;\;\pi^* \in \argmax\limits_{\pi\in \Acr}[g(.,Z+\ell\pi \sigma)-\ell \pi \sigma\theta].$$
\end{Theorem}
\begin{proof}
We first start by justifying that the expressions involved in the Theorem \ref{Th1} are well defined.
The results of conditional analysis explored in \cite{cheridito11} ensures the existence of the infimum and the supermum in (\ref{eq4}) and consequently the function $f$ is well defined  and also that the sets $\argmin\limits_{c\in \Cc} \Big(\alpha e^{-\int_0^. \delta_s ds } \ln (c)-\ell c\Big) $ et $\argmin\limits_{\pi\in \Ac}[g(.,Z+\ell\pi \sigma)-\ell \pi \sigma\theta]$ are non-empty. On the other hand the hypotheses on $h$ lead to the fact that $f$ is convex with quadratic variation in $z$ which ensures the existence and uniqueness of the solution of \ref{BSDE}.
 Moreover the conditions on $g$ implies that the function $f$ is quadratic in $z$ variable  which guarantees the existence  of solution for the BSDE (\ref{BSDE}).
 \\
Let's prove now the Theorem \ref{Th1}. The idea consists of constructing a family of processes $R^{x,\pi,c}$ which satisfies the conditions of the lemma \ref{lemma1}.
. We seek $R^{x,\pi,c}$ in the form
$$R_{t}^{x,\pi,c}:=\ell_t\ln( X_{t}^{x,\pi,c})+Y_t+\dint_0^t \alpha e^{-\int_0^u \delta_s ds } \ln (c_u X_{u}^{x,\pi,c})du$$ 
Where the process $Y$  satisfies the  quadratic BSDE (\ref{BSDE}) 
and the function $\ell$ is given by \ref{ell}. Note that the function $\ell$ is the unique  solution of the ordinary Cauchy problem 
\begin{equation}\label{cauchypb}
\forall 0 \leq t \leq T; \dfrac{d\ell_t}{dt} =-\alpha e^{-\int_0^t \delta_u du} \;\; \text{and} \;\; \ell_T=\bar{\alpha} e^{-\int_0^T \delta_u du}.
\end{equation}
Using It\^o formula, we have
\begin{align*}
dR_{t}^{x,\pi,c}&=-\alpha e^{-\int_0^t \delta_u du}\ln( X_{t}^{x,\pi,c}) dt +\ell_t (\pi_t \sigma_{t}dW_{t}-\pi_t \sigma_{t}\theta_{t}dt - c_tdt)- f(t,Z_t)dt+Z_tdW_t
\\& 
+\alpha e^{-\int_0^u \delta_s ds } \ln (c_t) dt +\alpha e^{-\int_0^u \delta_s ds } \ln (X_{t}^{x,\pi,c}) dt
\\&=  (-f(t,Z_t)-\ell_t \pi_t \sigma_{t}\theta_{t} - \ell_tc_t + \alpha e^{-\int_0^u \delta_s ds } \ln (c_t))dt + (\ell_t\pi_t \sigma_{t}+Z_t)dW_t
\\&= (-f(t,Z^{x,\pi,c}_t-\ell_t\pi_t \sigma_{t})-\ell_t \pi_t \sigma_{t}\theta_{t} - \ell_tc_t + \alpha e^{-\int_0^t \delta_s ds } \ln (c_t))dt + Z^{x,\pi,c}_tdW_t
\end{align*}
Where $Z_t^{x,\pi,c}:=\ell_t\pi_t \sigma_{t}+Z_t;\; 0 \leq t \leq T$. 
Therefore $R^{x,\pi,c}$ is solution of the BSDE
$$
\left\{\begin{array}{l}
d R_{t}^{x,\pi,c}=-F^{\pi,c}(t,Z^{x,\pi,c}_t)dt+Z^{x,\pi,c}_tdW_{t}, 
 \\
R_{T}^{x,\pi,c}:=\bar{\alpha}e^{-\int_0^T \delta_u du }\ln( X_{T}^{x,\pi,c})+\dint_0^T \alpha e^{-\int_0^u \delta_s ds } \ln (c_u X_{u}^{x,\pi,c})du.
\end{array}\right.
$$
Where
$$F^{\pi,c}(t,z)=f(t,z-\ell_t\pi_t \sigma_{t})+\ell_t \pi_t \sigma_{t}\theta_{t} + \ell_tc_t - \alpha e^{-\int_0^t \delta_s ds } \ln (c_t).$$
by definition of $f$ we have for all $t\in [0,+\infty[,z\in \R^d$ and for all $(\pi,c)\in \Acr\times \Ccr$
$$ f(t,z) \geq g(t,z+\ell_t\pi_t \sigma_{t})-\ell_t \pi_t \sigma_{t}\theta_{t}+ \alpha e^{-\int_0^t \delta_s ds } \ln (c_t)- \ell_tc_t ;\;\; P.a.s$$
which implies that for all $t\in [0,+\infty[,z\in \R^d$ and for all $(\pi,c)\in \Acr\times \Ccr$
$$F(t,z) \geq g(t,z); \;\;P.a.s.$$
And therefore, based on Proposition \ref{prop1}, $R^{x,\pi,c}$ is $g-$ supermartingale for all $(\pi,c)\in \Acr \times \Ccr.$ Moreover, let $(Y,Z)$ the unique solution of BSDE (\ref{BSDE}),  $c^* \in \argmin\limits_{c\in \Cc} \Big(\alpha e^{-\int_0^. \delta_s ds } \ln (c)-\ell c\Big)$ and  $\pi^* \in \argmin\limits_{\pi\in \Ac}[g(.,Z+\ell\pi \sigma)-\ell \pi \sigma\theta].$ Then , we have
$$
d R_{t}^{x,\pi^*,c^*}=-g(t,Z^{x,\pi^*,c^*}_t)dt+Z^{x,\pi^*,c^*}_tdW_{t}, 
 $$
 and so the process $ R^{x,\pi^*,c^*}$ is a $g-$ martingale. Which completes the proof of the theorem.
\end{proof}

\section{Example}
In this section we suppose that $h(x)=\dfrac{1}{2}\|x\|^2; \forall x \in \R^d$ which matches the entropic utility case. Then, we have $  h^*(x)=\dfrac{1}{2}\|x\|^2$ and therefore
$g(w,t,z)=-\dfrac{1}{2\beta }e^{\int_0^t \delta_u du} \|z\|^2$.
We also assume that the process $(\dfrac{\alpha e^{-\int_0^t\delta_udu}}{\ell_t})_{t\in[0,T]} \in \Ccr.$
\\For a process $q$ in $\mathcal{P}^{1 \times n}$, we denote by $\operatorname{dist}(q, P)$ the predictable process
$$
\operatorname{dist}(q, P):=\underset{p \in P}{\operatorname{essinf}}\|q-p\|,
$$
where ess inf denotes the greatest lower bound with respect to the $\nu \otimes \mathbb{P}$-a.e. order. It is shown in \cite{cheridito11} that there exists a process $p \in P$ satisfying $|q-p|=\operatorname{dist}(q, P)$ and that it is unique (up to $\nu \otimes \mathbb{P}$-a.e. equality) if $P$ is $\mathcal{P}$-convex. We denote the set of these processes by $\Pi_P(q)$.
\begin{Proposition}
The generator (\ref{eq4}) is given by
\begin{equation*}
 f(t,z) =
\alpha e^{-\int_0^t \delta_u du }\ln(\dfrac{\alpha e^{-\int_0^t\delta_udu}}{e\ell_t}) -\dfrac{1}{2\beta }e^{\int_0^t \delta_u du} \operatorname{dist}^2(z+\beta e^{-\int_0^t \delta_u du}\theta_t,\ell_t\sigma_t \Acr)+z\theta_t+\dfrac{\beta}{2}e^{-\int_0^t \delta_u du}\|\theta_t\|^2,
\end{equation*}
where the set $\ell\sigma \Acr:=\{\ell\sigma \pi ; \pi \in\Acr\}.$
\\The optimal strategy $(\pi^*, c^*)$ verifies:
$c^*_t=\dfrac{\alpha e^{-\int_0^t\delta_udu}}{\ell_t}$  and $\pi^* \in \dfrac{1}{\ell \sigma} \Pi_{\ell \sigma \Acr}(Z+\beta e^{-\int_0^. \delta_u du}\theta)$
\end{Proposition}
\begin{proof}
By studying the function $x \mapsto \alpha e^{-\int_0^t \delta_s ds } \ln (x)-\ell_t x $, it is clear that it reaches its maximum in $\dfrac{\alpha e^{-\int_0^t\delta_sds}}{\ell_t}$ and we have
$$\esssup\limits_{c \in \Ccr} \alpha e^{-\int_0^t \delta_s ds } \ln (c_t)-\ell_t c_t=\alpha e^{-\int_0^t \delta_s ds } ( \ln (\dfrac{\alpha e^{-\int_0^t\delta_udu}}{\ell_t})- 1)=\alpha e^{-\int_0^t \delta_s ds }  \ln (\dfrac{\alpha e^{-\int_0^t\delta_udu}}{e\ell_t}) $$
\begin{equation*}
\begin{split}
 &g \left(t, z+ \ell_t\pi_t\sigma_t)\right)-\ell_t\pi_t \sigma_{t} \theta_{t}
 \\&
 =-\dfrac{1}{2\beta }e^{\int_0^t \delta_u du} \|z+ \ell_t\pi_t\sigma_t\|^2-\ell_t\pi_t \sigma_{t} \theta_{t}
 \\&
 =-\dfrac{1}{2\beta }e^{\int_0^t \delta_u du} \|z+ \ell_t\pi_t\sigma_t\|^2-\ell_t\pi_t \sigma_{t} \theta_{t}
 \\&
 =-\dfrac{1}{2\beta }e^{\int_0^t \delta_u du} \|z+\beta e^{-\int_0^t \delta_u du}\theta_t + \ell_t\pi_t\sigma_t\|^2 +z\theta_t+\dfrac{\beta}{2}e^{-\int_0^t \delta_u du}\|\theta_t\|^2,
 \end{split}
\end{equation*}
and so, 
\begin{equation*}
\begin{split}
 & \esssup\limits_{\pi \in  \Acr } g \left(t, z+ \ell_t\pi_t\sigma_t)\right)-\ell_t\pi_t \sigma_{t} \theta_{t}
 \\&
  =-\dfrac{1}{2\beta }e^{\int_0^t \delta_u du} \essinf\limits_{\pi \in  \Acr } \|z+\beta e^{-\int_0^t \delta_u du}\theta_t + \ell_t\pi_t\sigma_t\|^2 +z\theta_t+\dfrac{\beta}{2}e^{-\int_0^t \delta_u du}\|\theta_t\|^2
  \\&
  =-\dfrac{1}{2\beta }e^{\int_0^t \delta_u du} \operatorname{dist}^2(z+\beta e^{-\int_0^t \delta_u du}\theta_t,\ell_t\sigma_t \Acr) +z\theta_t+\dfrac{\beta}{2}e^{-\int_0^t \delta_u du}\|\theta_t\|^2.
 \end{split}
\end{equation*}
\end{proof}
 \section{Conclusion}
 In this paper, we studied the robust utility maximization problem in the logarithmic utility framework and in an incomplete market with general constraint. Using nonlinear martingale we characterized the optimal strategies using quadratic BSDE. Although our study is limited to the case of logarithmic utility, it generalizes the work of Cheridito et al. \cite{cheridito10} in the  case of robust utility maximization and and Jiang et al. \cite{Jiang16}.  in the case of  model with consumption and general constraints.

\end{document}